\documentclass{theoretics}

\usepackage{amsfonts}

\newcommand\calS{{\mathcal{S}}}
\newcommand\calA{{\mathcal{A}}}
\newcommand\calG{{\mathcal{G}}}

\newcommand{\R}{{\mathbb{R}}}
\newcommand{\Q}{{\mathbb{Q}}}
\newcommand\ceil[1]{{\lceil#1\rceil}}
\newcommand{\lp}{\mathrm{lp}}

\DeclareMathOperator{\E}{\mathbb{E}}

\title{A Note on Approximating Weighted Nash Social Welfare with Additive Valuations}

\ThCSauthor[Nanjing,NCL]{Yuda Feng}{yudafeng@smail.nju.edu.cn}[0009-0004-7461-2285]

\ThCSauthor[Nanjing,NCL]{Shi Li}{shili@nju.edu.cn}[0000-0001-9140-9415]
\ThCSaffil[Nanjing]{School of Computer Science, Nanjing University, Nanjing, Jiangsu, China}  
\ThCSaffil[NCL]{New Cornerstone Science Laboratory}

\ThCSshortnames{Y. Feng, S. Li}  
\ThCSshorttitle{Approximating Weighted Nash Social Welfare with Additive Valuations}
\ThCSyear{2025}
\ThCSarticlenum{17}
\ThCSreceived{Oct 30, 2024}
\ThCSrevised{May 27, 2025}
\ThCSaccepted{Jun 4, 2025}
\ThCSpublished{Aug 14, 2025}
\ThCSdoicreatedtrue
\ThCSkeywords{Nash Social Welfare, Configuration LP, Approximation Algorithms}

\ThCSthanks{This paper was invited from ICALP 2024 \cite{Feng024}. The work of YF and SL was supported by the State Key Laboratory for Novel Software Technology, and the New Cornerstone Science Laboratory.}

\begin{document}

\maketitle

\begin{abstract}
    We give the first $O(1)$-approximation for the weighted Nash Social Welfare problem with additive valuations. The approximation ratio we obtain is $e^{1/e} + \epsilon \approx 1.445 + \epsilon$, which matches the best known approximation ratio for the unweighted case \cite{BKV18}.  
    
    Both our algorithm and analysis are simple. We solve a natural configuration LP for the problem, and obtain the allocation of items to agents using the Shmoys-Tardos rounding algorithm developed for  unrelated machine scheduling problems \cite{ST93}. In the analysis, we show that the approximation ratio of the algorithm is at most the worst gap between the Nash social welfare of the optimum allocation and that of an EF1 allocation, for an unweighted Nash Social Welfare instance with identical additive valuations. This was shown to be at most $e^{1/e} \approx 1.445$ by Barman, Krishnamurthy and Vaish\ \cite{BKV18}, leading to our approximation ratio. 
\end{abstract}

\section{Introduction}

In the weighted (or asymmetric) Nash Social Welfare problem with additive valuations, we are given a set $\calA$ of $n$ agents, and a set $\calG$ of $m$ indivisible items. Every agent $i \in \calA$ has a weight $w_i \geq 0$ such that $\sum_{i \in \calA}w_i = 1$. There is a value $v_{ij} \in \R_{\geq 0}$ for every $i \in \calA$ and $j \in \calG$.  The goal of the problem is to find an allocation $\sigma: \calG \to \calA$ of items to agents so as to maximize the following weighted Nash Social Welfare of $\sigma$: 
\begin{equation*}
    \prod_{i \in \calA} \bigg(\sum_{j \in \sigma^{-1}(i)}v_{ij}\bigg)^{w_i}.
\end{equation*}
In the case where all $w_i$'s are equal to $\frac1n$, we call the problem the unweighted (or symmetric) Nash Social Welfare problem.

Allocating resources in a fair and efficient manner among multiple agents is a fundamental problem in computer science, game theory, and economics, with applications across diverse domains \cite{KN79, You94, BT96, RW98, Mou04, BT05,  Rot15, BCE16}. The weighted Nash Social Welfare function is a notable objective that balances efficiency and fairness. The unweighted (or symmetric) objective was independently proposed by different communities \cite{Nas50, Kel97, Var74}, and later the study has been extended to the weighted case \cite{HS72, Kal77}. Since then it has been used in a wide range of applications, including bargaining theory \cite{LV07, CM10, Tho86}, water allocation \cite{HLZ13, DHY18}, and climate agreements \cite{YvIW17}.

The unweighted Nash Social Welfare problem with additive valuations is proved to be NP-hard by Nguyen, Nguyen, Roos and Rothe \cite{NNRR12}, and APX-hard by Lee \cite{Lee15}.  Later the hardness of approximation was improved to $\sqrt{8/7} \approx 1.069$ by Garg, Hoefer and Mehlhorn \cite{GHM23}, via a reduction from Max-E3-Lin-2.

On the positive side, Cole and Gkatzelis \cite{CG15} gave a $(2e^{1/e} + \epsilon\approx 2.889 + \epsilon)$-approximation using a spending restricted market equilibrium relaxation. The ratio was improved by Cole, Devanur, Gkatzelis, Jain, Mai, Vazirani and Yazdanbod \cite{CDG17} to $2$ using a tight analysis, and by Anari, Oveis Gharan, Saberi and Singh \cite{AGS17} to $e$ via a connection of the problem to real stable polynomials. Both papers formulated convex programming (CP) relaxations of the problem. In particular, \cite{CDG17} showed that the optimum solution to their CP corresponds to the spending-restricted market equilibrium defined in \cite{CG15}.  The state-of-the-art result for the problem is a combinatorial $(e^{1/e}+\epsilon\approx 1.445 + \epsilon)$-approximation algorithm due to Barman, Krishnamurthy and Vaish \cite{BKV18}. They showed that when all the valuations of agents are identical, any allocation that is envy-free up to one item (EF1) is $e^{1/e}$-approximate.  Their approximation result then follows from a connection between the non-identical and identical valuation settings they established. 

All the results discussed above are for the unweighted case.  For the weighted case with agent weights $w \in [0, 1]^\calA, |w|_1 = 1$, Brown, Laddha, Pittu and Singh \cite{BLP24} presented a $5 \cdot \exp(2 \cdot D_{\text{KL}}\Big(w || \frac{\vec{1}}{n})\Big) = 5\cdot \exp(2\log n + 2\sum_{i \in \calA} w_i \log w_i)$ approximation algorithm, where $D_{\text{KL}}$ denotes the KL divergence of two distributions. Their result is based on the CP from \cite{CDG17}, generalized to the weighted setting. 
\smallskip

The additive valuation setting is a special case of the submodular valuation setting, which is another important setting studied in the literature. In this setting, instead of a $v_{ij}$ value for every $ij$ pair, we are given a monotone submodular function $v_i: 2^{\calG} \to \R_{\geq 0}$ for every agent $i \in \calA$. Till the end of the section, our goal is to find an allocation $\sigma: \calG \to \calA$ so as to maximize $\prod_{i \in \calA} \Big(v_i(\sigma^{-1}(i))\Big)^{w_i}$. 
A bulk of the previous work has focused on the unweighted case; that is, $w_i = \frac1n$ for all $i \in \calA$. For this case, \cite{GKK23} proved a hardness of $e/(e-1)\approx 1.5819$ using a reduction from Max-3-Coloring; this is better than the $1.069$ hardness for the additive valuation case. 

On the positive side, Li and Vondr\'{a}k \cite{LV21} extended the techniques of \cite{AGS17}, to obtain an $e^3/(e-1)^2$-approximation algorithm for the unweighted Nash Social Welfare problem for a large family of submodular valuations, including coverage functions and linear combinations of matroid rank functions. Later, Garg, Husi\'{c}, and V\'{e}gh \cite{GHV21} considered a family of submodular functions called Rado functions, and gave an $O(1)$-approximation for this family using the matching theory and convex program techniques.  Li and Vondr\'{a}k \cite{LV22} developped the first $O(1)$-approximation for general submodular functions, with an approximation ratio of $380$.  Recently, Garg, Husi\'{c}, Li, V\'{e}gh and Vondr\'{a}k \cite{GHL23} presented an elegant $4$-approximation local search algorithm for the problem, which is the current best approximation result for the problem.  All the results discussed above are for the unweighted case. For the weighted case, \cite{GHL23} gave an $O(nw_{\max})$-approximation, where $w_{\max} = \max_{i \in \calA} w_i$. 

Using our configuration LP idea, Feng, Hu, Li and Zhang \cite{FHLZ25} recently developed a $(233+\epsilon)$-approximation for the weighted case, which is the first $O(1)$-approximation for the weighted Nash Social Welfare problem with submodular valuations.

Recently, the problem has been studied in an even more general setting, namely, the subadditive valuation setting.  Dobzinski, Li, Rubinstein and Vondr\'{a}k \cite{DLR23} gave an $O(1)$-approximation for the unweighted Nash Social Welfare problem in this setting under the demand oracle model.

\subsection{Our Result and Techniques} 

In this note, we give the first $O(1)$-approximation algorithm for the weighted Nash Social Welfare problem with additive valuations:
\begin{theorem}
    \label{thm:main}
    For any $\epsilon > 0$, there is a deterministic $(e^{1/e} + \epsilon \approx 1.445 + \epsilon)$-approximation algorithm for the weighted Nash Social Welfare problem with additive valuations, with running time polynomial in the size of the input and $\frac1\epsilon$.
\end{theorem}

Our approximation ratio of $e^{1/e} + \epsilon$ matches the best ratio for the unweighted case due to Barman, Krishnamurthy and Vaish \cite{BKV18}.  In contrast, the ratio given by \cite{BLP24} is $5\cdot \exp(2 \cdot D_{\text{KL}}(w || \frac{\vec{1}}{n}))$, which could be polynomial in $n$. 

Our algorithm is based on a natural configuration LP for the problem, which has not been studied before to the best of our knowledge.  The configuration LP contains a $y_{i, S}$ variable for every agent $i$ and subset $S$ of items, indicating if the set of items $i$ gets is $S$ or not. We show that the configuration LP can be solved in polynomial time to any precision, despite having an exponential number of variables.  Once we obtain the LP solution, we define $x_{ij}$ for every $i \in \calA$ and $j \in \calG$ to be the fraction of $j$ assigned to $i$. 

To round $x$ into an integral solution, it is more convenient to focus on a randomized algorithm first, which uses the Shmoys-Tardos rounding algorithm \cite{ST93} developed for unrelated machine scheduling problems. 
For every agent $i$, we break the fractional items assigned to $i$ into groups from the most valuable to the least, each containing 1 fractional item. The rounding algorithm maintains marginal probabilities, and the requirement that $i$ gets exactly one item from each group (except for the last one, from which $i$ gets at most one item).  In the analysis for each agent $i$, we construct an instance of the unweighted Nash Social Welfare problem with \emph{identical} additive valuations, that involves many copies of the agent $i$, along with two allocations $\calS$ and $\calS'$ to the instance. $\calS$ corresponds to the LP solution, and $\calS'$ corresponds to the randomized solution given by the rounding algorithm. Thanks to the condition that every group contains one item, the solution $\calS'$ is envy-free up to one item (EF1), defined as follows.
    \begin{definition}
        Given an instance of the unweighted Nash Social Welfare problem with agents $\calA$, items $\calG$, and identical additive valuation $v: \calG \to \R_{\geq 0}$ for all agents, an allocation $\sigma: \calG \to \calA$ is said to be envy-free up to one item (EF1),  if for every two distinct agents $i, i'$ with $\sigma^{-1}(i') \neq \emptyset$, there exists some $j \in \sigma^{-1}(i')$, such that $v(\sigma^{-1}(i') \setminus j) \leq v(\sigma^{-1}(i))$.
    \end{definition}

Barman, Krishnamurthy and Vaish \cite{BKV18} showed the following result on the quality of EF1 allocations:
    \begin{theorem}[\cite{BKV18}]
        \label{thm:EF1}
        For the unweighted Nash Social Welfare problem with identical additive valuations, any EF1-allocation is an $e^{1/e}$-approximate solution. 
    \end{theorem} 

The result implies that the Nash social welfare of $\calS'$ is at least $e^{-1/e}$ times that of $\calS$. This in turn proves that the expected weighted Nash Social Welfare of the solution returned by our rounding algorithm is at least the exponential of the value of the LP solution.  To de-randomize the algorithm, we break the fractional matching between groups and items into a convex combination of polynomial number of integral matchings, and choose the best matching in the combination.

\section{\texorpdfstring{$(e^{1/e} + \epsilon)$}{e to the 1/e + epsilon}-Approximation Using Configuration LP}
At the beginning of the algorithm, we check if there is an allocation with positive weighted Nash Social Welfare, using the bipartite matching algorithm. If not, we terminate the algorithm immediately by returning any allocation. From now on, we assume the optimum weighted Nash Social Welfare value is positive.  We describe the configuration LP in Section~\ref{sec:LP} and the rounding algorithm in Section~\ref{sec:rounding}. The analysis is given in Section~\ref{sec:analysis}.

\subsection{The Configuration LP} \label{sec:LP}

For convenience, for any value function $v: \calG \to \R_{\geq 0}$, we define $v(S):= \sum_{j \in S}v_j$ for every $S \subseteq \calG$ to be the total value of items in $S$ according to the value function $v$. In the integer program corresponding to the configuration LP, for every $i \in \cal A$ and $S \subseteq \calG$, we have a variable $y_{i, S} \in \{0, 1\}$ indicating if the set of items assigned to $i$ is $S$ or not. We relax the integer constraint to obtain the following configuration LP: 
\begin{align}
    \max \sum_{i \in \calA, S \subseteq  \calG} w_i \cdot y_{i, S} \cdot \ln v_i(S) \qquad \text{s.t.} \tag{Conf-LP}\label{CLP}
\end{align} \vspace*{-10pt}
\begin{align}
    \sum_{i \in \calA, S \ni j} y_{i, S} &\leq 1 &\quad &\forall j \in \calG \label{LPC:j-assigned-once}\\
    \sum_{S \subseteq \calG} y_{i, S} &= 1 &\quad &\forall i \in \calA \label{LPC:i-gets-one-set}\\
    y_{i, S} &\geq 0 &\quad &\forall i \in \calA, S\subseteq \calG
\end{align}

It is convenient for us to consider the natural logarithm of the Nash social welfare function as the objective, which is $\sum_{i \in \calA} w_i \cdot \ln v_i (\sigma^{-1}(i))$. This leads to the objective in \eqref{CLP}.  \eqref{LPC:j-assigned-once} requires that every item $j$ is assigned to at most one agent, and \eqref{LPC:i-gets-one-set} requires that every agent $i$ is assigned one set of items. 

The configuration LP has an exponential number of variables, but it can be solved within an additive error of $\ln(1 + \epsilon)$ for any $\epsilon > 0$, in time polynomial in the size of the instance  and $\frac1\epsilon$. 
\begin{theorem}
    \label{thm:solve-lp}
    For any $\epsilon > 0$, there is an algorithm that     
    outputs a valid solution $(y_{i, S} \in \Q_{\geq 0})_{i \in \calA, S \subseteq \calG}$ to \eqref{CLP} whose value is at least the optimum value of the LP minus $\ln(1+\epsilon)$, represented using a list of the non-zero entries. The running time of the algorithm is polynomial in the input size and $\frac1\epsilon$.
\end{theorem}
We defer the proof of Theorem~\ref{thm:solve-lp} to Section~\ref{sec:solving-LP}. 
Notice that we are considering the logarithm of Nash social welfare, and the typical $(1+\epsilon)$-multiplicative factor becomes an additive error of $\ln(1 + \epsilon)$.

\subsection{The Rounding Algorithm} \label{sec:rounding}
    From now on, we assume we have obtained a vector $y$ from solving the LP, described using a list of non-zero coordinates; the value of $y$ to \eqref{CLP} is at least the optimum value minus $\ln(1+\epsilon)$. We can assume \eqref{LPC:j-assigned-once} holds with equalities: $\sum_{i \in \calA, S \ni j}y_{i, S} = 1$ for every $j \in \calG$. Then we let $x_{ij} = \sum_{S \ni j} y_{i, S}$ for every $i \in \calA$ and $j \in \calG$. So $\sum_{i \in \calA} x_{ij} = 1$ for every $j \in \calG$.

    In this paragraph, we fix an agent $i$ and break the fractional items assigned to $i$ into a set $G_i$ of groups, each containing 1 fractional item. They are created in non-increasing order of values, as in the Shmoys-Tardos algorithm for unrelated machine scheduling problems. That is, the first group contains the 1 fractional most valuable items assigned to $i$, the second group contains the 1 fractional most valuable items assigned to $i$ after removing the first group, and so on. Formally, we sort the items in $\calG$ in non-increasing order of $v_{ij}$ values, breaking ties arbitrarily. Let $p_i = \ceil{\sum_{j \in \calG} x_{ij}}$. Then we can find vectors $g^1, g^2, \cdots, g^{p_i} \in [0, 1]^{\calG}$ satisfying the following properties:
    \begin{enumerate}[label = (P\arabic*)]
        \item For every $t \in [1, p_i-1]$, we have $|g^t|_1 = 1$; for $t = p_i$, we have $|g^t|_1 = \sum_{j \in \calG}x_{ij} - (p_i - 1) \in (0, 1]$.
        \item $\sum_{t = 1}^{p_i}g^t_j = x_{ij}$ for every $j \in \calG$.
        \item For every $1 \leq t < t' \leq p_i$, and two items $j, j'$ such that $j$ appears before $j'$ in the ordering, it cannot happen that $g^t_{j'} > 0$ and $g^{t'}_j > 0$.
    \end{enumerate}
    It is easy to see that $g^1, g^2, \cdots, g^{p_i}$ are uniquely determined by the three conditions. We say each $g^t$ is a group. Let $G_i = \{g^1, g^2, \cdots, g^{p_i}\}$ be the set of all groups constructed for this agent $i$. 
    \medskip

    Now we take all agents $i$ into consideration and let $G = \uplus_{i \in \calA}G_i$ be the set of all groups constructed.\footnote{It is possible that two groups from different sets $G_i$ and $G_{i'}$ have the same vector representation. So we treat $G$ as a multi-set and we assume we know which set $G_i$ each group $g \in G$ belongs to.} The representations of groups give a fractional matching between the groups $G$ and items $\calG$: an item $j$ is matched to a group $g \in [0, 1]^{\calG}$ with a fraction of $g_j$.  Then each item is matched to an extent of 1, and every group $g$ is matched to an extent of $|g|_1$. So a group is matched to an extent of 1 if it is not the last group for an agent, and at most 1 otherwise.  
    
    We can efficiently partition the fractional matching between $G$ and $\calG$ into a convex combination of polynomial number of (partial-)matchings.  If we randomly choose a matching from the convex combination, the following property holds. 

    \begin{enumerate}[label = ($\star$)]
        \item For every group $g \in G$ and item $j \in \calG$, we have $\Pr[j\text{ is matched to }g] = g_j$.
    \end{enumerate}

    Each matching in the combination naturally gives us an allocation of items to agents: If an item $j \in \calG$ is matched to some group $g \in G_i$, then we assign $j$ to $i$.  Our algorithm simply outputs the best allocation from all matchings in the combination.

\subsection{The Analysis} \label{sec:analysis}
    It is more convenient to analyze the following randomized rounding algorithm: randomly choose a matching from the convex combination, and output the allocation corresponding to the matching. Clearly, the deterministic algorithm can only perform better.

    Let $S_i$ be the set of items assigned to $i$ by the randomized rounding algorithm. By ($\star$) we know that the probability that $j$ is assigned to $i$ is precisely $x_{ij}$. ($\star$) implies that an item $j \in \calG$ is matched with probability 1. If a group $g$ has $|g|_1 = 1$, then it is matched with probability 1.

    With Theorem~\ref{thm:EF1} on the quality of EF1 allocations in the setting of identical agents, we prove the following key lemma:
    \begin{lemma}
        For every $i \in \calA$, we have
        \begin{align*}
            \E\big[\ln v_i(S_i)\big] \geq \sum_{S \subseteq \calG} y_{i, S}\cdot \ln v_i(S) \ -\ \frac1e.
        \end{align*}
    \end{lemma}
    \begin{proof}
        Throughout the proof, we fix the agent $i$. Let $\Delta > 0$ be an integer, so that every $y_{i, S}$ is an integer multiple of $1/\Delta$, and the probability that $S_i = S$ for any $S$ is also an integer multiple of $1/\Delta$.\footnote{We have that all $y_{i, S}$ values are rational numbers. Under this condition, it is  easy to guarantee that the probabilities are rational numbers.} We consider an instance of the unweighted Nash Social Welfare problem with identical additive valuations. In the instance, there are $\Delta$ copies of the agent $i$, and $\Delta x_{ij}$ copies of every item $j \in \calG$; so all the agents are identical. The $y = (y_{i, S})_{S \subseteq \calG}$ vector gives us an allocation $\calS$ to the instance: For every $S \subseteq \calG$, there are exactly $\Delta y_{i,S}$ agents who get a copy of $S$.  Notice that this is a valid solution, as $\sum_{S} y_{i, S} = 1$ and $\sum_{S \ni j} y_{i, S} = x_{ij}$ for every item $j$.
        
        The Nash Social Welfare of the allocation $\calS$ is 
        \begin{align*}
            \left(\prod_{S \subseteq \calG} v_i(S)^{\Delta y_{i, S}}\right)^{1/\Delta}  = \prod_{S \subseteq \calG} v_i(S)^{y_{i, S}}.
        \end{align*}
        
        The distribution for $S_i$ also corresponds to an allocation $\calS'$ of items to agents: For every $S \subseteq \calG$, there are $\Delta \cdot \Pr[S_i = S]$ agents who get a copy of $S$. Again, this is a valid solution as $\sum_{S} \Pr[S_i = S] = 1$ and $\sum_{S \ni j} \Pr[S_i = S] = \E[S_i \ni j] = x_{ij}$.
        
        The Nash Social Welfare of the allocation $\calS'$ is 
        \begin{align*}
            \left(\prod_{S \subseteq \calG} v_i(S)^{\Delta \Pr[S_i = S]}\right)^{1/\Delta} = \prod_{S \subseteq \calG} v_i(S)^{\Pr[S_i = S]}.
        \end{align*}
        
        A crucial property for the solution $\calS'$ is that it is EF1. Indeed, if $\Pr[S_i = S] > 0$ for some $S$, then $S$ contains exactly one item from each group in $G_i$ except for the last one, from which $S$ contains at most one item. Also, the items in the groups $G_i$ are sorted by (P3).  So if there are two sets $S$ and $S'$ in the support of the distribution for $S_i$, and we remove the most valuable item from $S'$, then $S$ beats $S'$ item by item. 
        
        Therefore, by Theorem~\ref{thm:EF1}, we know that the Nash Social Welfare of $\calS'$ is at least $e^{-1/e}$ times that of the optimum allocation for the instance, which is at least that of $\calS$. That is, 
        \begin{align*}
            \prod_{S \subseteq \calG} v_i(S)^{\Pr[S_i = S]} \geq e^{-1/e} \cdot \prod_{S \subseteq \calG} v_i(S)^{y_{i, S}}.
        \end{align*}
        Taking logarithm on both sides gives the lemma. 
    \end{proof}

    Applying the lemma for every $i \in \calA$ and using linearity of expectation, we have 
    \begin{align*}
        \E\left[\sum_{i \in \calA}w_i \cdot \ln v_i(S_i)\right] \geq \sum_{i \in \calA, S \subseteq \calG} w_i \cdot y_{i, S} \cdot \ln v_i(S) - \frac1e.
    \end{align*}
    We used that $\sum_{i \in \calA} w_i = 1$. 
    By the convexity of exponential function, we have 
    \begin{align*}
        \E\left[\prod_{i \in \calA}v_i(S_i)^{w_i}\right] \geq e^{-1/e} \cdot \exp\left(\sum_{i \in \calA, S \subseteq \calG} w_i \cdot y_{i, S} \cdot \ln v_i(S_i)\right) \geq e^{-1/e} \cdot \frac{\text{opt}}{1 + \epsilon},
    \end{align*}
    where $\text{opt}$ is the weighted Nash Social Welfare of the optimum allocation, and the second inequality used that the value of our solution $y$ to \eqref{CLP} is at least its optimum value minus $\ln(1+\epsilon)$.
    By scaling $\epsilon$ down by an absolute constant at the beginning, we can make the right side to be at least $\frac{\text{opt}}{e^{1/e} + \epsilon}$. As we argued, the deterministic algorithm will output an allocation whose weighted Nash Social Welfare is at least $\frac{\text{opt}}{e^{1/e} + \epsilon}$. This finishes the proof of Theorem~\ref{thm:main}.

\section{Solving Configuration LP within an Additive Error of \texorpdfstring{$\ln (1 + \epsilon)$}{log (1 + epsilon)}: Proof of Theorem~\ref{thm:solve-lp}} \label{sec:solving-LP}

    In this section, we prove Theorem~\ref{thm:solve-lp}. By scaling valuation functions, we assume $v_{ij} = 0$ or $v_{ij} \geq 1$ for every $i \in \calA, j \in \calG$; this does not change the instance. Let $\lp$ be the value of the \eqref{CLP}. As we assumed there is an allocation with positive value, we have $\lp \neq -\infty$.

    We assume $\epsilon > 0$ is upper bounded by a sufficiently small constant; otherwise, we take $\epsilon$ to be the constant. We say an algorithm is efficient if its running time is polynomial in the size of the input instance and $\frac1\epsilon$. We assume that we are given a number $o < \lp$. Our goal is to find a valid solution to \eqref{CLP} with rational coordinates, whose value is at least $o - \frac\epsilon 4 - \ln\big(1 + \frac\epsilon 4\big)$. At the end of the section, we show that this is sufficient. \medskip

    We consider the dual of \eqref{CLP}, with the objective replaced by a constraint. 
    \begin{align}
        \sum_{j \in \calG}\alpha_j + \sum_{i \in \calA} \beta_i &\leq o \label{LPC:dual-obj}\\
        \sum_{j \in S} \alpha_j + \beta_i &\geq w_i \cdot \ln v_i(S) &\quad &\forall i \in \calA, S \subseteq \calG \label{LPC: dual-iS}\\
        \alpha_j &\geq 0 &\quad &\forall j \in \calG \label{LPC:dual-alpha} 
    \end{align}
    As we assumed $o < \lp$, the dual LP (\ref{LPC:dual-obj}-\ref{LPC:dual-alpha}) is infeasible. Abusing the terminology slightly, we say a vector $(\alpha \in \R_{\geq 0}^\calG, \beta \in \R^{\calA})$ is an optimum solution to the dual LP if it minimizes the left-side of \eqref{LPC:dual-obj}, subject to \eqref{LPC: dual-iS} and \eqref{LPC:dual-alpha}. By linear programming duality, an optimum solution $(\alpha, \beta)$ has $\sum_{j \in \calG}\alpha_j + \sum_{i \in \calA} \beta_i=\lp$. 

    To check the feasibility of dual LP, we need an approximate separation oracle for \eqref{LPC: dual-iS}. This can be done using dynamic programming: 
    \begin{lemma}
        \label{lemma:oracle}
        Given $(\alpha \in \R_{\geq 0}^\calG, \beta \in \R^\calA)$ that does not satisfy \eqref{LPC: dual-iS}, there is an efficient oracle that finds some $i \in \calA$ and $S' \subseteq \calG$ satisfying 
    \begin{align}
        \sum_{j \in S'} \alpha_j + \beta_i < w_i \cdot \ln \Big(\big(1+\frac\epsilon2\big)v_i(S')\Big). \label{inequ:S'}
    \end{align}
  
    \end{lemma}

    \begin{proof}
        Suppose \eqref{LPC: dual-iS} is not satisfied for $i \in \calA$ and $S \subseteq \calG$. We can guess the agent $i$, and the item $j^* \in S$ with the largest $v_{ij^*}$. Then we discard the items $j$ with $v_{ij} > v_{ij^*}$. Let $\calG'$ be the set of remaining items.  For every $j \in \calG'$, we round down $v_{ij}$ to the nearest integer multiple of $\frac{\epsilon\cdot v_{ij^*}}{2m}$; let $\bar v_{ij}$ be the rounded value. As $\bar v_{ij}$ values are integer multiples of $\frac{\epsilon\cdot v_{ij^*}}{2m}$ not exceeding $v_{ij^*}$, we can afford to guess $\bar V := \bar v_i(S)$ as it has only $O(\frac{m^2}{\epsilon})$ possibilities. 

        We solve the following knapsack covering problem: find a set $S' \subseteq \calG'$ such that $\bar v_i(S') \geq \bar V$ so as to minimize $\sum_{j \in S'} \alpha_j$.  As the $\bar v_{ij}$ values are of the form $\frac{z \cdot \epsilon\cdot v_{ij^*}}{2m}$ with $z \in \Big[0, \frac{2m}{\epsilon}\Big]$ being integers, the problem can be solved efficiently and exactly using dynamic programming. As the set $S$ is a valid solution to the knapsack covering instance, the solution $S'$ returned satisfies $\bar v_i(S') \geq \bar V = \bar v_i(S)$ and $\sum_{j \in S'} \alpha_j \leq \sum_{j \in S} \alpha_j$. So, 
    \begin{align*}
        v_i(S') \geq \bar v_i(S') \geq \bar v_i(S) \geq v_i(S) - m\cdot \frac{\epsilon\cdot v_{ij^*}}{2m} = v_i(S) - \frac{\epsilon\cdot v_{ij^*}}{2} \geq v_i(S) - \frac{\epsilon}{2} \cdot v_i(S').
    \end{align*}
    Therefore $v_i(S) \leq \big(1 + \frac\epsilon2\big) v_i(S')$. Then \eqref{inequ:S'} follows from that \eqref{LPC: dual-iS} is not satisfied for $i, S$.

    We need to guess $i, j^*$ and $\bar V$; there are $O(\frac{nm^3}{\epsilon})$ possibilities for the combination. For a fixed $i, j^*$ and $\bar V$, the dynamic programming runs in $O(\frac{m^3}{\epsilon})$-time. So, overall, the running time of the oracle is $\mathrm{poly}\big(n, m, \frac1\epsilon\big)$.\footnote{We remark that we did not try to optimize the running time.} Finally as \eqref{inequ:S'} can be verified easily for any given $i$ and $S'$, incorrect guesses will not cause any issue. 
    \end{proof}

    To run the ellipsoid method, we need some bounds on $\alpha_j$ and $\beta_i$ values. This is not straightforward as the $\beta_i$ values can be negative.  We define $v_{\max} := \max_{i \in \calA, j \in \calG} v_{ij}$.
    \begin{lemma}
        \label{lemma:bound-alpha-beta}
        There is an optimum solution $(\alpha \in \R_{\geq 0}^\calG, \beta \in \R^{\calA})$ to the dual LP such that $\alpha_j \in [0, \ln (m v_{\max}^2)]$ for every $j \in \calG$ and $|\beta_i| \leq \ln (mv_{\max}^2)$ for every $i \in \calA$.
    \end{lemma}

    \begin{proof}
        We modify constraint~\eqref{LPC: dual-iS} in the dual LP  (\ref{LPC:dual-obj}-\ref{LPC:dual-alpha}), so that we only consider the constraint for pairs $(i \in \calA, S \subseteq \calG)$ satisfying $v_{ij} > 0$ for every $j \in S$. This does not change the dual LP as constraint~\eqref{LPC: dual-iS} for a pair $(i, S)$ is implied by the constraint for the pair $(i, \{j \in S: v_{ij}> 0\})$. Let $\big(\alpha \in \R_{\geq 0}^\calG, \beta \in \R^{\calA}\big)$ an optimum solution to the dual LP with the smallest $\sum_{j \in \calG}\alpha_j$; notice that $\sum_{j \in \calG}\alpha_j + \sum_{i \in \calA} \beta_i = \lp$. Let $\Pi$ be the set of $(i, S)$ pairs for which \eqref{LPC: dual-iS} is tight for the given $(\alpha, \beta)$, after the modification to constraint~\eqref{LPC: dual-iS}.
            
        We argue that if $\beta_i < 0$ and $|S| \geq 2$, then $(i, S) \notin \Pi$. Assume towards the contradiction that some $(i, S) \in \Pi$ has $\beta_i < 0$ and $|S| \geq 2$.  
        
        Every $j \in S$ has $v_{ij} \geq 1$ by the assumption made at the beginning of the section.  We have $\alpha_j + \beta_i\geq w_i\cdot \ln v_{ij}$ for every $j \in S$ as \eqref{LPC: dual-iS} holds for $(i, \{j\})$. Summing up the inequalities over all $j \in S$ gives us $\sum_{j \in S}\alpha_j + |S|\beta_i \geq w_i \cdot \sum_{j \in S}\ln v_{ij} \geq w_i \ln v_i(S)$. Therefore $\sum_{j \in S}\alpha_j + \beta_i > w_i \ln v_i(S)$ as $|S|\geq 2$ and $\beta_i < 0$, a contradiction.
    
        Let $\calA'$ be the set of agents $i$ with $\beta_i < 0$. Let $E = \big\{(i, j): i \in \calA', (i, \{j\}) \in \Pi\big\}$. We prove that there is a matching in the bipartite graph $H:= (\calA' \cup \calG, E)$ covering all agents $\calA'$. Otherwise, there is a set $\calA'' \subseteq \calA'$ of agents with $|N_H(\calA'')| < |\calA''|$, where $N_H(\calA'')$ is the set of neighbors of $\calA''$ in $H$. Then consider the following operation: decreasing $\beta_i$ for every $i \in \calA''$ by $\delta$ and increasing $\alpha_j$ for every $j \in N_H(\calA'')$ by $\delta$, for a small enough $\delta > 0$. This decreases $\sum_{j \in \calG}\alpha_j + \sum_{i \in \calA}\beta_i$ as $|C \cap \calG|\geq |C \cap \calA|$ without violating \eqref{LPC: dual-iS} and \eqref{LPC:dual-alpha}, contradicting that $(\alpha, \beta)$ is optimum. 

        Focus on any connected component $C$ of $H$. We prove that some $j \in C \cap \calG$ has $\alpha_j \leq \ln(m v_{\max})$. Assume this is not the case; that is, every $j \in C \cap \calG$ has $\alpha_j > \ln (m v_{\max})$. Consider the following operation: decreasing $\alpha_j$ by $\delta$ for every $j \in C \cap \calG$ and increasing $\beta_i$ by $\delta$ for every $i \in C \cap \calA$, for a small enough $\delta > 0$. This operation does not break constraint~\eqref{LPC: dual-iS} for any $i$ with $\beta_i \geq 0$, as the right side of \eqref{LPC: dual-iS} is at most $\ln (mv_{\max})$. It does not break the constraint for any $i$ with $\beta_i < 0$ as $\delta$ is sufficiently small.  Moreover, the operation decreases $\sum_{j \in \calG}\alpha_j$ without increasing $\sum_{j \in \calG}\alpha_j + \sum_{i \in \calA}\beta_i$ as $|C \cap \calG|\geq |C \cap \calA|$, contradicting the way we choose $(\alpha, \beta)$. 
        
        Therefore, every connected component $C$ of $H$ contains some $j \in C \cap \calG$ with $\alpha_j \leq \ln(m v_{\max})$. Recall that every $(i, j) \in E$ has $\beta_i < 0$ and $\beta_i + \alpha_j = w_i \cdot \ln v_{ij}$. This is equivalent to $\alpha_j  = -\beta_i + w_i \cdot \ln v_{ij}$. There is a simple path in $H$ connecting every vertex $\calA' \cup \calG$ to an item $j$ with $\alpha_j \leq \ln (mv_{\max})$. So, $\alpha_j$ for every $j \in \calG$ (respectively, $-\beta_i$ for every $i \in \calA'$) is at most $\ln (mv_{\max}) + \sum_{i \in \calA'} w_i \ln v_{\max} \leq \ln (mv_{\max}) + \ln v_{\max} = \ln(m v_{\max}^2)$. Finally, for $i \in \calA \setminus \calA'$, we have $\beta_i \geq 0$. Clearly we have $\beta_i \leq w_i\ln (mv_{\max}) \leq \ln (mv_{\max})$ since other decreasing $\beta_i$ to $w_i\ln (mv_{\max})$ will not break any constraint, violating the optimality of $(\alpha, \beta)$. 
    \end{proof}

    With the lemma, we start the ellipsoid method with the smallest ellipsoid containing the cuboid $[0, mv^2_{\max}]^{\calG} \times [-mv^2_{\max}, mv^2_{\max}]^{\calA}$.
    We terminate the procedure when the volume of the ellipsoid is less than that of the cuboid $\Big[0, \frac{\epsilon}{4(n+m)}\Big]^{n + m}$. The algorithm terminates in $\mathrm{poly}\big(n, m, \log \frac 1\epsilon, \log v_{\max}\big)$ iterations and is thus efficient.\medskip
    
    We claim the following LP is infeasible:
    \begin{align}
        \sum_{j \in \calG}\alpha_j + \sum_{i \in \calA} \beta_i &\leq o - \frac\epsilon 4 \label{LPC:dual-obj-1}\\
        \sum_{j \in S'} \alpha_j + \beta_i &\geq w_i \cdot \ln \Big(\big(1+\frac\epsilon2\big)v_i(S')\Big) &\quad &\forall (i, S') \text{ used by the ellipsoid method}\label{LPC: dual-iS-1}\\
        \alpha_j &\geq 0 &\quad &\forall j \in \calG \label{LPC:dual-alpha-1} 
    \end{align}
   
    Otherwise, assume some $(\alpha, \beta)$ is a feasible solution to the above LP. Then the cuboid $\Big\{(\alpha, \beta) + \delta : \delta \in \big[0, \frac{\epsilon}{4(n+m)}\big]^{\calG \cup \calA}\Big\}$ would be still contained in the final ellipsoid, contradicting our termination condition.

    We consider the weighted Nash Social Welfare instance where all $v_{ij}$ values are scaled up by $1+\frac\epsilon 2$, and \eqref{CLP} to the instance. By solving the LP restricted to the variables $y_{i, S'}$ corresponding to constraints \eqref{LPC: dual-iS-1} (that is, we fix the other variables to $0$), we obtain a solution $y$ whose value is at least $o - \frac\epsilon 4$ w.r.t the scaled instance. The LP solver runs in polynomial time as the number of variables $y_{i, S'}$ now is polynomial.  So, the value of the solution $y$ to \eqref{CLP} w.r.t the original instance is at least $o - \frac\epsilon 4 - \sum_{i \in \calA, S \subseteq \calG} y_{i, S}w_i\ln\big(1+\frac\epsilon 4\big) = o - \sum_{i} w_i\ln\big(1+\frac\epsilon 4\big) = o - \frac\epsilon 4- \ln\big(1+\frac\epsilon 4\big)$. Moreover, if $y$ is a vertex solution, then it is rational as all the coefficients in the LP constraints are in $\{0, 1\}$.

    Finally, we show how to handle the assumption that we are given the parameter $o$. By only allowing every agent to get one item, we can obtain an $m$-approximation for the weighted Nash Social Welfare instance. Therefore, we can construct a list of $O\Big(\frac{\log m}{\epsilon}\Big)$ values of $o$, such that for some value $o$ in the list, we have $\lp \in (o, o + \frac\epsilon 4]$. For every $o$ in the list, we run the above algorithm and return the best solution constructed. In particular, for the $o$ with $\lp \in (o, o + \frac\epsilon 4]$, the solution has value at least $o - \frac\epsilon 4 - \ln\big(1+\frac\epsilon 4\big) \geq \lp - \frac\epsilon 4 - \frac\epsilon 4 - \ln\big(1+\frac\epsilon 4\big) \geq \lp - \ln(1+\epsilon)$ for a sufficiently small $\epsilon$. Therefore, the best solution returned has value at least $\lp - \ln(1+\epsilon)$; this finishes the proof of Theorem~\ref{thm:solve-lp}.


\printbibliography

\end{document}